\documentclass[10pt]{amsart}

\usepackage{hyperref}

\usepackage{amssymb}
\usepackage{amsmath}

\let\newpf\proof \let\proof\relax

\def\bm{\begin{matrix}}
\def\em{\end{matrix}}

\newcommand{\bt}{\begin{thm}}
\newcommand{\et}{\end{thm}}

\newcommand{\bl}{\begin{lemma}}
\newcommand{\el}{\end{lemma}}

\newcommand{\beq}{\begin{eqnarray}}
\newcommand{\eeq}{\end{eqnarray}}

\def\be{\begin{equation}}
\def\ee{\end{equation}}

\def\ba{{\begin{align}}}
\def\ea{{\end{align}}}

\def\0{{\mathbf 0}}

\newtheorem{thm}{Theorem}[section]

\newtheorem{lemma}[thm]{Lemma}

\theoremstyle{remark}
\newtheorem{rem}{Remark}[section]

\numberwithin{equation}{section}

\def \bn {\hfill \\ \smallskip\noindent}

\theoremstyle{definition}

\newtheorem{definition}{Definition}[section]
\def\proof{\bn {\bf Proof.} }

\def\note#1
{\marginpar

{\tiny $\leftarrow$
\par
\hfuzz=20pt \hbadness=9000 \hyphenpenalty=-100 \exhyphenpenalty=-100
\pretolerance=-1 \tolerance=9999 \doublehyphendemerits=-100000
\finalhyphendemerits=-100000 \baselineskip=6pt
#1}\hfuzz=1pt}

\newcommand{\dist}{\operatorname{dist}}

\newcommand{\Q}{{\mathbb Q}}
\newcommand{\R}{{\mathbb R}}
\newcommand{\T}{{\mathbb T}}

\newcommand{\Z}{{\mathbb Z}}

\def\B0{{\bold{0}}}

\catcode`\@=12

\def\Empty{}
\newcommand\oplabel[1]{
  \def\OpArg{#1} \ifx \OpArg\Empty {} \else
  	\label{#1}
  \fi}

\newcommand{\comm}[1]{}
\newcommand{\comment}[1]{}

\begin{document}

\title{Pure point spectrum for the Maryland model: a constructive proof}

\author{Svetlana Jitomirskaya and Fan Yang}

\thanks{}

\begin{abstract}
We develop a constructive method to prove and study pure point spectrum for the
Maryland model with Diophantine frequencies.
\end{abstract}

\maketitle

\section{Introduction}

The Maryland model
is a  
discrete self-adjoint 
Schr\"{o}dinger operator on $\ell^2(\mathbb{Z})$ of the form
 \begin{equation}\label{marylandmodel}
(H_{\lambda,\alpha,\theta}u)_n=u_{n+1}+u_{n-1}+\lambda\tan\pi(\theta+n\alpha)u_n.
\end{equation}
where $\lambda \in \R$ is called the coupling, $\alpha \in \R$ the frequency, and
$\theta \in \R$ is the phase.  In this paper we will assume $\alpha \in
\R \setminus \Q$, $\lambda>0$. Let $\Theta\triangleq\frac{1}{2}+\alpha
\Z+\Z$. Clearly, for $\theta\in \Theta$ the operator is not well
defined. From now on when we say ``all $\theta$'', we will mean ``$\theta\notin \Theta$''. 

Maryland model is a linear version of the quantum kicked rotor in the momentum space, originally proposed by Grempel,
Fishman and Prange \cite{fgp}.  As an exactly solvable example of the
family of incommensurate models, it attracts continuing interest in
physics, e.g. \cite{berry}, \cite{fishman2010}, \cite{GKDS2014}. 

For Diophantine frequencies $\alpha,$ Maryland model has localization:
pure point
spectrum with exponentially decaying eigenfunctions, for all $\theta$ \cite{fp1984,simm}. In fact, it was recently
shown in \cite{maryland} that $\sigma_{pp}(H_{\lambda,\alpha,\theta})$
can be characterized arithmetically, in an exact way, for all parameters. Namely, an index $\delta(\alpha,\theta) \in [-\infty,\infty]$ was  introduced
in \cite{maryland} and it was shown that  $\sigma_{pp}(H_{\lambda,\alpha,\theta})=
\{E: L_\lambda(E)\geq \delta(\alpha,\theta)\},$ while $\sigma_{sc}(H_{\lambda,\alpha,\theta})=\overline{\{E:L_{\lambda}(E) <\delta (\alpha, \theta) \}}$ where $L_\lambda(E)$
is the Lyapunov exponent, see (\ref{LEmaryland}) (which for the Maryland model
does not depend on $\alpha,\theta).$

 It should be noted that all the
proofs of localization so far, including those mentioned above, as
well as
the original physics paper \cite{fgp}, have been indirect: based on a
Cayley transform that reduced the eigenvalue problem to solving 
certain explicit cohomological equation. In this paper we present a different
approach, by proving exponential decay of all polynomially bounded
solutions directly. The eigenfunctions of the Maryland model are, as 
a result of indirect analysis, known exactly, yet the formulas don't
make it easy to make conclusions about their behavior, which is quite
interesting, with transfer matrices satisfying certain exact
renormalization \cite{fs}. The advantage of our approach is that it
provides a completely different and rather promising way to study the solutions, for example, their asymptotics and
various other features. For example, the Maryland eigenfunctions are expected,
through numerics, to have hierarchical structure driven by the
continued fraction expansion of the frequency, and our method 
has the potential to be developed to study that, as was recently done for the almost Mathieu
operator in \cite{jl1}. Moreover, we expect to be able to also use
this method to study some features of solutions/spectral
measures/quantum dynamics in
the singular continuous regime,\footnote{First studied in physics
  literature in \cite{wave}.} as was done for the almost Mathieu
operator in \cite{jlt}.  In general, Maryland model, being exactly solvable, has been a very useful
laboratory in the field of quasiperiodic operators, as a source of
both general conjectures and counterexamples. The possibility of
direct analysis of Maryland eigenfunctions, presented in this
manuscript,  provides a new very important tool to this laboratory.

Here we show how the argument works in the
simplest, that is Diophantine, case. We will develop a more delicate method  and apply it to study the
full localization region in the upcoming work \cite{hjy}.


Let $\frac{p_n}{q_n}$ be the continued fraction approximants of $\alpha\in \R \setminus \Q$.
Let
\begin{equation}
\beta=\beta(\alpha)=\limsup_{n\rightarrow\infty}\frac{\ln q_{n+1}}{q_n},
\end{equation} 
We call $\alpha$ {\it Diophantine} if $\beta(\alpha)=0$.





A formal solution $\phi(x)$ of $H\phi=E\phi$ is called a generalized eigenfunction if $\phi$ is a non-trivial solution, and $|\phi(x)| \leq C(1+|x|)$ for some constant $0<C< \infty$.

Our main result is:
\begin{thm}\label{main}
For Diophantine $\alpha$ and any $\theta$, any generalized
eigenfunction of $H_{\lambda, \alpha, \theta}$ decays exponentially.
\end{thm}

\begin{rem}
By  Schnol's theorem \cite{berez,rui}, Theorem \ref{main} is
  equivalent to the statement that $H_{\lambda, \alpha, \theta}$ has
  pure point spectrum with exponentially decaying eigenfunctions, a
  known result, as already mentioned. We
  choose to formulate it the way we do to underscore the new method of
  proof.
  \end{rem}

There has been a number of papers lately with constructive proofs of localization
with arithmetic conditions  for the almost Mathieu and extended
Harper's model \cite{kjs,AJ1,liuyuan,jl1,jl2}, all dealing
with $\cos$ potential. Here we show that the method of \cite{j} can
also be developed in an even simpler way to treat the Maryland model.

\section{Preliminaries}

Let $\mathbb{T}=\mathbb{R}/\mathbb{Z}$ be the one dimensional
torus. For $x \in \mathbb{R}$, let $\|x\|_{\mathbb{T}}= \dist
(x,\mathbb{Z})$. Later, we will also sometimes write $\|x\|$ for $\|x\|_{\T}$.

\subsection{Cocycles and Lyapunov exponents}
For $\alpha\in \R$, and $A:\T \to M_2(\mathbb{C})$ or  $A:\T \to
M_2(\mathbb{C})/\{\pm 1\}$ Borel measurable satisfying 
\begin{align}\label{cocyclecondition}
\log|\det(A(\theta))| \in L^1(\T),
\end{align}
we call the pair $(\alpha, A)$ a {\it cocycle} understood as a linear skew-product on $\T\times \mathbb{C}^2$ (or $\T\times \mathbb{PC}^2$) defined by 
\begin{align*}
(\alpha,A): (\theta, v) \mapsto (\theta+\alpha, A(\theta) \cdot v )
\end{align*}
The {\it Lyapunov exponent} of $(\alpha, A)$ is defined by
\begin{align}\label{defLE}
L(\alpha, A)=\lim_{k\rightarrow \infty}\frac{1}{k}\int_{\T}\ln{\|A(\theta+(k-1)\alpha) \cdots A(\theta+\alpha) A(\theta)\|}\mathrm{d}\theta.
\end{align}
If we consider the eigenvalue equation of the Maryland model $H_{\lambda, \alpha, \theta}\phi=E\phi$, then any solution can be reconstructed via the following relation
\begin{align*}
\left(
\begin{matrix}
\phi(k+1)\\
\phi(k)
\end{matrix}
\right)
=
D(\theta+k\alpha, E)
\left(
\begin{matrix}
\phi(k)\\
\phi(k-1)
\end{matrix}
\right),
\end{align*}
where 
\begin{align}\label{transferA}
D(\theta, E)=
\left(
\begin{matrix}
E-\lambda\tan{\pi \theta}\ &-1\\
1                          &0
\end{matrix}
\right).
\end{align}
Iterating this process, we will get
\begin{align*}
\left(
\begin{matrix}
\phi(k)\\
\phi(k-1)
\end{matrix}
\right)
=
D_k(\theta, E)
\left(
\begin{matrix}
\phi(0)\\
\phi(-1)
\end{matrix}
\right),
\end{align*}
where
\begin{align*}
\left\lbrace
\begin{matrix}
D_k(\theta, E)=D(\theta+(k-1)\alpha, E)\cdots D(\theta+\alpha, E)D(\theta, E)\ \mathrm{for}\ k\geq 1,\\
D_0(\theta, E)=\mathrm{Id},\\
D_k(\theta, E)=(D_{-k}(\theta+k\alpha, E))^{-1}\ \mathrm{for}\ k\leq -1.
\end{matrix}
\right.
\end{align*}
The pair $(\alpha, D)$ is the Schr\"odinger cocycle corresponding to
the Maryland model. We will denote the Lyapunov exponent $L(\alpha, D(\cdot, E))$ by $L(E)$. $D_k$ is called the {\it $k$-step transfer matrix}. It was shown in \cite{fgp} that
\begin{align}\label{LEmaryland}
e^{L(E)}+e^{-L(E)}=\frac{\sqrt{(2+E)^2+\lambda^2}+\sqrt{(2-E)^2+\lambda^2}}{2}.
\end{align}

Note that the cocycle $(\alpha, D)$ is singular because it
contains $\tan{\pi\theta}$. As it is more convenient to work with
non-singular cocycles, we introduce
\begin{align}\label{defnonsingular}
F(\theta, E)=\cos{\pi \theta}\cdot D(\theta, E)=
\left(
\begin{matrix}
E\cos{\pi \theta}-\lambda\sin{\pi \theta}\ &-\cos{\pi \theta}\\
\cos{\pi \theta}                          &0
\end{matrix}
\right).
\end{align}
Note that $F$ is an $M_2(\R)/\{\pm 1\}$ cocycle (that is, defined up to a sign).
We will denote its Lyapunov exponent $L(\alpha, F(\cdot, E))$ by
$\tilde{L}(E)$. Clearly by (\ref{defLE}) and the fact that $\int_{\T} \ln|\cos\pi \theta| d\theta=-\ln2$, we have the following relation between $L(E)$ and $\tilde{L}(E)$.
\begin{align}\label{LEtildeLE}
\tilde{L}(E)=L(E)-\ln{2}.
\end{align}

The following control of the norm of the transfer matrix of a uniquely
ergodic continuous cocycle by the Lyapunov exponent is well known.
\begin{lemma}\label{upperbounds}$\mathrm{(}$e.g. \cite{Furman, JMavi}$\mathrm{)}$
Let $(\alpha, M)$ be a continuous cocycle, then for any $\epsilon>0$, for $|k|$ large enough,
\begin{align*}
\|M_k(\theta)\|\leq e^{|k|(L(\alpha, M)+\epsilon)}\ \mathrm{for}\ \mathrm{any}\ \theta\in\T.
\end{align*}
\end{lemma}
\begin{rem}\label{partf}
Considering 1-dimensional continuous cocycles, a corollary of Lemma \ref{upperbounds} is that if $g$
is a continuous function such that $\ln |g| \in L^1(\T)$, 
then for any $\epsilon > 0$, and $b-a$ sufficiently large, 
$$|\prod_{j=a}^b g (\theta+j\alpha) | \leq e^{(b-a+1)(\int_{\T} \ln|g| \mathrm{d}\theta+\epsilon)}.$$
In particular, we obtain upper bound of $|\prod_{j=a}^{b} \cos \pi (\theta+j \alpha)|$ as follows:
\begin{equation}\label{cosln2control}
|\prod_{j=a}^{b} \cos \pi (\theta+j \alpha)| \leq e^{(b-a+1)(-\ln2+\epsilon)}.
\end{equation}
\end{rem}

\subsection{A closer look at the transfer matrix}
If we consider the Schr\"odinger cocycle $(\alpha, D(\theta, E))$, it turns out $D_k(\theta, E)$ has the following expression
\begin{equation}\label{PinA}
D_k(\theta,E)=
\left(
\begin{array}{cc}
P_k(\theta,E) & -P_{k-1}(\theta+\alpha,E) \\
P_{k-1}(\theta,E) & -P_{k-2}(\theta+\alpha,E)
\end{array}
\right),
\end{equation}
where $P_k(\theta,E)=\det{[(E-H_{\theta})|_{[0,k-1]}]}$.

Let $\tilde{P}_k(\theta, E):\R/2\Z\to\R$ be defined as $\tilde{P}_k(\theta, E)=\prod_{j=0}^{k-1} \cos \pi (\theta+j \alpha)\cdot P_k(\theta, E)$.  Then clearly
\begin{align}\label{tildePinF}
F_k(\theta, E)=
\left(
\begin{array}{cc}
\tilde{P}_k(\theta,E) & -\tilde{P}_{k-1}(\theta+\alpha,E)\cos{\pi\theta} \\
\tilde{P}_{k-1}(\theta,E)\cos{\pi (\theta+(k-1)\alpha)} & -\tilde{P}_{k-2}(\theta+\alpha,E)\cos{\pi\theta}\cos{\pi(\theta+(k-1)\alpha)}
\end{array}
\right).
\end{align}
By the fact that $F$ is continuous and by  (\ref{tildePinF})  and Lemma \ref{upperbounds}, we have the following upper bound on $\tilde{P}_k$.
\begin{lemma}\label{upperbddtildeP}
For any $\epsilon>0$ for $|k|$ large enough,
\begin{align}
|\tilde{P}_k(\theta, E)|\leq e^{(\tilde{L}(E)+\epsilon)|k|}\ \mathrm{for}\ \mathrm{any}\ \theta\in \T.
\end{align}
\end{lemma}

\subsection{Solution and Green's function}

We use $G_{[x_1,x_2]}(E)(x,y)$ for the Green's function $(H-E)^{-1}(x,y)$ of the operator $H_{\lambda,\alpha,\theta}$ restricted to the interval $[x_1,x_2]$ with zero boundary conditions at $x_1-1$ and $x_2+1$. We will omit $E$ when it is fixed throughout the argument.

Let $\phi$ be a solution to $H\phi=E\phi$ and let $[x_1, x_2]$ be an interval containing $y$. We have
\begin{align}\label{expand}
\phi(y)=-G_{[x_1, x_2]}(x_1, y)\phi(x_1-1)-G_{[x_1, x_2]}(x_2, y)\phi(x_2+1).
\end{align}

By Cramer's rule, we have the following connection between the determinants $\tilde{P}_k$ and Green's function:
\begin{align}\label{PkG1}
|G_{[x_1, x_2]}(x_1, y)|&=\frac{|P_{x_2-y}(\theta+(y+1)\alpha)|}{|P_{x_2-x_1+1}(\theta+x_1 \alpha)|} \\ \notag
                                     &=\frac{|\tilde{P}_{x_2-y}(\theta+(y+1)\alpha)|}{|\tilde{P}_{x_2-x_1+1}(\theta+x_1\alpha)|}\prod_{j=x_1}^{y}|\cos{\pi (\theta+j\alpha)}|,\\ \notag
\end{align} 
\begin{align}\label{PkG2}                                  
|G_{[x_1, x_2]}(x_2, y)|&=\frac{|P_{y-x_1}(\theta+x_1 \alpha)|}{|P_{x_2-x_1+1}(\theta+x_1 \alpha)|} \\ \notag
                                     &=\frac{|\tilde{P}_{y-x_1}(\theta+x_1\alpha)|}{|\tilde{P}_{x_2-x_1+1}(\theta+x_1\alpha)|}\prod_{j=y}^{x_2}|\cos{\pi (\theta+j\alpha)}|. \notag
\end{align}

\subsection{Regular and singular points}
\begin{definition}\label{regular}
A point $y \in\Z$ will be called $(m,h)$-regular if there exists an interval $[x_1,x_2]$, $x_2=x_1+h-1$, containing $y$, such that,
\begin{align}
|G_{[x_1,x_2]}(x_i, y)|<e^{-m|y-x_i|},\ |y-x_i|\geq \frac{1}{100}h\ \mathrm{for}\  i=1,2.
\end{align}
otherwise, $y$ will be called $(m,h)$-singular.
\end{definition}

\subsection{Rational approximations}
Let $\{\frac{p_n}{q_n}\}$ be continued fraction approximants of $\alpha$, then
\begin{align}
&\frac{1}{2q_{n+1}}\leq \|q_n\alpha\|\leq \frac{1}{q_{n+1}}, \\
\mathrm{and} \      &\|k\alpha\| \geq \|q_n \alpha\| \  \mathrm{for} \  0< |k| <q_{n+1}.
\end{align}
If $\alpha$ is Diophantine, then for $n$ large enough, we have
\begin{align}\label{qnalpha}
\|q_n\alpha\|\geq e^{-\epsilon q_n}.
\end{align}
\subsection{Trigonometric product}

The following Lemma from \cite{AJ1} gives a useful estimate of products appearing in our analysis.
\begin{lemma} \label{lana}\cite{AJ1}
Let $\alpha\in \R\setminus \Q $,\ $\theta\in\R$ and $0\leq j_0 \leq q_{n}-1$ be such that 
$$\mid \cos \pi(\theta+j_{0}\alpha)\mid = \inf_{0\leq j \leq q_{n}-1} \mid \cos \pi(\theta+j\alpha)\mid ,$$
then for some absolute constant $C>0$,
$$-C\ln q_{n} \leq \sum_{j=0,j\neq j_0}^{q_{n}-1} \ln \mid \cos \pi (\theta+j\alpha) \mid+(q_{n}-1)\ln2 \leq C\ln q_n$$
\end{lemma}

\section{Key lemmas}
\subsection{Average lower bound of $\tilde{P}_k$}
We now give the following average lower bound of $\tilde{P}_k$:
\begin{lemma}\label{subharmonic}
For $k$ large enough we have
\begin{align}\label{averagelower}
\frac{1}{k}\int_{\mathbb{T}}\ln |\tilde{P}_k(\theta)| \mathrm{d}\theta =\frac{1}{k} \int_{\mathbb{T}} \ln |\tilde{P}_k(2\theta)| \mathrm{d}\theta \geq L(E)-\ln2.
\end{align}
\end{lemma}
This lemma will be proved in Section \ref{Proof of Lemma averagelower}.

\subsection{Lagrange interpolation for $\tilde{P}_k$}

An important observation that makes our analysis possible is

\begin{lemma}

$ \tilde{P}_k(\theta)/{\cos^k{\pi\theta}}$ 
can be expressed as a polynomial of degree $k$ in $\tan\pi\theta$, namely,
\begin{equation}\label{Phpolyg}
\frac{\tilde{P}_k(\theta)}{(\cos{\pi\theta})^k}  \triangleq g_k(\tan\pi\theta).
\end{equation}

\end{lemma}

\begin{rem} 
While $ \tilde{P}_k(\theta)$ is a function on $\R/2\Z$, $ \tilde{P}_k(\theta)/{\cos^k{\pi\theta}}$ is
  a function on $\R/\Z.$
\end{rem}

\begin{proof} An induction, using that $P_k(\theta)=P_{k-1}(\theta)(E-\tan\pi(\theta+(k-1)\alpha))-P_{k-2}(\theta)$. $\hfill{} \Box $
\end{proof}

By the Lagrange interpolation formula, for any set of $k+1$ distinct $\theta_i$'s in $(-1/2,1/2)$,
\begin{equation*}
g_k(\tan\pi \theta)=\sum_{i=0}^{k}g_k(\tan \pi \theta_i) \frac{\prod_{l \neq i}(\tan \pi \theta-\tan\pi\theta_l)}{\prod_{l\neq i}(\tan\pi\theta_i-\tan\pi\theta_l)}.
\end{equation*}
Thus we have the following convenient representation
\begin{align}
\tilde{P}_k(\theta)=(\cos\pi \theta)^k g_k(\tan\pi \theta)
&=\sum_{i=0}^{k}\tilde{P}_k(\theta_i)\frac{\prod_{l \neq i}\tan \pi \theta-\tan\pi\theta_l}{\prod_{l\neq i}\tan\pi\theta_i-\tan\pi\theta_l}\cdot  \frac{\cos^k {\pi \theta}}{\cos^k{\pi\theta_i}} \notag\\
&=\sum_{i=0}^{k}\tilde{P}_k(\theta_i)\prod_{l\neq i} \frac{\sin\pi(\theta-\theta_l)}{\sin\pi(\theta_i-\theta_l)}. \label{tildePlagrange}
\end{align}

\subsection{Uniformity}
\begin{definition}
We say that the set $\{\theta_1,..., \theta_{k+1}\}$ is $\epsilon$-uniform if
\begin{align}\label{defuniform}
\max_{\theta\in [0, 1]}\max_{i=0,..., k} \prod_{l\neq i} \frac{|\sin\pi(\theta-\theta_l)|}{|\sin\pi(\theta_i-\theta_l)|}<e^{k\epsilon}.
\end{align}
\end{definition}
Note that this differs from the definition of uniformity used in \cite{AJ1,liuyuan,jl1,jl2}.
For a fixed $k$, choose the largest $q_n$ such that $\frac{1}{25}q_n \leq |k|$. We will assume $k\geq 0$. 
We define $I_1$ and $I_2$ differently in the following two cases:

\subparagraph{Case 1}
If $\frac{1}{25}q_n \leq k <q_n$, let $h=2q_n$, and set
\begin{align}
\left\lbrace
\begin{matrix}
I_1=[k-2q_n-[\frac{2q_n}{100}]+1, k-q_n-[\frac{2q_n}{100}]],\\
I_2=[k-[\frac{2q_n}{100}]-q_n+1, k-[\frac{2q_n}{100}]].
\end{matrix}
\right.
\end{align}


\subparagraph{Case 2}
If $q_n \leq k < \frac{1}{25}q_{n+1}$,
there exists the smallest positive integer $s$ such that 
\begin{align}
(2s-1)q_n \leq k < (2s+1)q_n.
\end{align} 
Let $h=2sq_n$ and set
\begin{align}
\left\lbrace
\begin{matrix}
I_1=[-2sq_n+[\frac{2sq_n}{100}]+1, -sq_n+[\frac{2sq_n}{100}]],\\
I_2=[k-[\frac{2sq_n}{100}]-sq_n+1, k-[\frac{2sq_n}{100}]].
\end{matrix}
\right.
\end{align}
For both cases, $I_1\cup I_2$ consists of $h$ points. From now on we fix $0<\epsilon<\frac{L(E)}{600}$.
We will show that 
\begin{lemma}\label{uniformcase}
For all k sufficiently large,
$\{\theta+l\alpha\}_{l\in I_1\cup I_2}$ is $3\epsilon$-uniform.
\end{lemma}
The proof will be given in Section \ref{proofofuniform}.

\subsection{Upper bound on $\tilde{P}_{h-1}$ on $I_1$}

We will show that $\tilde{P}_{h-1}$ cannot be large on $I_1$, namely,
\begin{lemma}\label{I1Psmall}
For $h$ large enough, for any $x_1\in I_1$, we have $|\tilde{P}_{h-1}(\theta+x_1\alpha)|<e^{h(\tilde{L}-4\epsilon)}$.
\end{lemma}
\begin{proof}
We will prove by contradiction.
Without loss of generality, assume $\phi(0) \neq 0$. Suppose there exists $x_1 \in I_1$ such that $|\tilde{P}_{h-1}(\theta+x_1\alpha)| \geq e^{h(\tilde{L}-4\epsilon)}$. By (\ref{expand}) and definition of $I_1$, we have
\begin{align}\label{expand0}
|\phi(0)|=|G_{[x_1, x_2]}(x_1, 0)\phi(x_1-1)+G_{[x_1, x_2]}(x_2, 0)\phi(x_2+1)|,
\end{align}
where $x_1 < 0 < x_2=x_1+h-2$, $|x_i| \geq \frac{h}{100}$. 

Using the fact that the numerators of Green's functions can be bounded uniformly by Lemma \ref{upperbddtildeP}, and using (\ref{LEtildeLE}), (\ref{cosln2control}), (\ref{PkG1}), (\ref{PkG2}), we can get upper bounds for the following Green's functions:
\begin{align*}
|G_{[x_1, x_2]}(x_1, 0)|  &=\frac{|\tilde{P}_{x_2}(\theta+\alpha)|}{|\tilde{P}_{h-1}(\theta+x_1\alpha)|}\prod_{j=x_1}^{0}|\cos{\pi (\theta+j\alpha)}|,\\ \notag
                                     &\leq 2e^{-|x_1|L+5h\epsilon} \\ \notag
                                     &\leq 2e^{-\frac{h}{100}(L-500\epsilon)} \rightarrow 0\\ \notag                                     
|G_{[x_1, x_2]}(x_2, 0)| &=\frac{|\tilde{P}_{-x_1}(\theta+x_1\alpha)|}{|\tilde{P}_{h-1}(\theta+x_1\alpha)|}\prod_{j=0}^{x_2}|\cos{\pi (\theta+j\alpha)}|. \\ \notag
                                     &\leq 2e^{-x_2 L+5h\epsilon} \\ 
                                     &\leq 2e^{-\frac{h}{100}(L-500\epsilon)} \rightarrow 0 
\end{align*}
For large $h$, this contradicts our assumption $\phi(0) \neq 0$. Therefore, for any $x_1 \in I_1$, we have $|\tilde{P}_{h-1}(\theta+x_1\alpha)|<e^{h(\tilde{L}-4\epsilon)}$.
$\hfill{} \Box$
\end{proof}

\subsection{Regularity of $k$}
\begin{lemma}\label{Plowbdd}
For large $|k|$, there exists $x_1 \in I_2$, such that $|\tilde{P}_{h-1}(\theta+x_1\alpha)| \geq e^{h(\tilde{L}-4\epsilon)}$
\end{lemma}
\begin{proof}
Consider the Lagrange interpolation of $\tilde{P}_{h-1}(\theta)$ (\ref{tildePlagrange}). By Lemmas \ref{uniformcase}, \ref{I1Psmall}, we obtain that it is impossible to have $|\tilde{P}_{h-1}(\theta+x_1 \alpha)|< e^{h(\tilde{L}-4\epsilon)} $ for all $x_1 \in I_1 \cup I_2$, as it will contradict Lemma \ref{averagelower}. 
Thus we can conclude that there must exist an $x_1 \in I_2$, such that $|\tilde{P}_{h-1}(\theta+x_1 \alpha)| \geq e^{h(\tilde{L}-4\epsilon)}$.   $\hfill{} \Box$
\end{proof} 

The existence of such $x_1$ leads to the following lemma.
\begin{lemma}\label{notwosingular} 
For $|k|\in \Z$ large enough,  $k$ is $(L-500\epsilon, h-1)$-regular.
\end{lemma}
\begin{rem}
By our choice of $h$, we have $\frac{2}{3}|k| \leq h \leq 50|k|$
\end{rem}

\begin{proof}
Since $x_2=x_1+h-2$, again the numerators of both expressions in (\ref{PkG1}) and (\ref{PkG2}) can be controlled using (\ref{cosln2control}) and Lemma \ref{upperbddtildeP}, that is 
\begin{align*}
&|\tilde{P}_{x_2-k}(\theta+(k+1)\alpha)| \leq e^{(\tilde{L}(E)+\epsilon)(x_2-k+1)},\\
&|\tilde{P}_{k-x_1}(\theta+x_1 \alpha)| \leq e^{(\tilde{L}(E)+\epsilon)(k-x_1+1)},\\
&\prod_{j=x_1}^{k}|\cos{\pi (\theta+j\alpha)}|\leq e^{(-\ln{2}+\epsilon)(k-x_1+1)},\\
&\prod_{j=k}^{x_2}|\cos{\pi (\theta+j\alpha)}|\leq e^{(-\ln{2}+\epsilon)(x_2-k+1)}.
\end{align*}
The regularity is then immediate from the definition and Lemma \ref{Plowbdd}.  

$\hfill{} \Box$
\end{proof}

\section{Proof of Theorem \ref{main}}
Applying Lemma \ref{notwosingular}, we obtain that for large $|k|$ there exists an interval $[x_1,x_2]$, $x_2=x_1+h-2$, containing $k$, such that
\begin{equation}
|G_{[x_1,x_2]}(x_i, k)| \leq e^{-(L-500\epsilon)|k-x_i|},\ |k-x_i|\geq \frac{h}{100}\ \mathrm{for}\ i=1,2.
\end{equation}
Thus by (\ref{expand}),
\begin{align*}
|\phi(k)|\leq &|G_{[x_1,x_2]}(x_1, k)\phi(x_1-1)|+|G_{[x_1,x_2]}(x_2, k)\phi(x_2+1)| \\
        \leq&e^{-(L-500\epsilon)|k-x_1|}|\phi(x_1-1)|+e^{-(L-500\epsilon)|k-x_2|}|\phi(x_2+1)| \\
        \leq&e^{-\frac{h}{100}(L-500\epsilon)}(1+C|x_1-1|+1+C|x_2+1|) \\
          < &e^{-|k|(\frac{L}{150}-4\epsilon)}    
\end{align*} $\hfill{} \Box$

\section{Proof of Lemma \ref{uniformcase}}\label{proofofuniform}
For any $i\in \mathbb{Z}$, let $\theta_i:= \theta+i\alpha$.
\subsection{Case 1:}

$$\frac{1}{25}q_n \leq k <q_n,\ h=2q_n$$
We divide the $2q_n$ points into two intervals: $T_1,T_2$, each interval containing $q_n$ points.
Fix any $i$. Let $|\sin\pi(\theta_i-\theta_{l_j})|$ be the minimal one
of $|\sin\pi(\theta_i-\theta_l)|$ in each $T_j, j=1,2$. Without loss
of generality, assume $i \in T_1$. Then $l_1=i$ and for any $x\in\T$
we have
\begin{align} \label{sin} 
  &\prod_{l\neq i} \frac{|\sin\pi(x-\theta_l)|}{|\sin\pi(\theta_i-\theta_l)|} \\ \notag
=&\exp\{\sum_{l\neq i} \ln|\sin\pi(x-\theta_l)|-\sum_{l\neq i} \ln|\sin\pi(\theta_i-\theta_l)|\}
\end{align}

We estimate the  two parts separately. First, using
Lemma \ref{lana},
\begin{align*}
       &\sum_{l\neq i} \ln|\sin\pi(x-\theta_l)| \\
     =&\sum_{j=1}^2\  \sum_{l\in T_j, l\neq i} \ln |\sin\pi(x-\theta_l)| \\
      <& 2q_n(-\ln2+\epsilon).
\end{align*}
The maximum distance between $i$ and $l_2$ is $2q_n$. However, it may exceed $q_{n+1}$. In this case, $q_{n+1}$ must be equal to $q_n+q_{n-1}$.
Thus we have the following estimates, using Lemma \ref{lana} and(\ref{qnalpha}):
\begin{align*}
       &\sum_{l\neq i} \ln{|\sin\pi(\theta_i-\theta_l)|}\\
    = &\sum_{l \in T_1,l \neq i}\ln |\sin\pi(\theta_i-\theta_l)| +\sum_{l \in T_2}\ln |\sin\pi(\theta_i-\theta_l)| \\
\geq&2(-C\ln q_n -(q_n-1)\ln2)+\ln \|q_{n+1}\alpha\| \\
\geq&2q_n(-\ln2-\epsilon)-\epsilon q_{n+1}\\
\geq&2q_n(-\ln2-2\epsilon).
\end{align*}
Besides, if $2q_n$ does not exceed $q_{n+1}$, the estimate will be:
\begin{align*}
       &\sum_{l\neq i} \ln{|\sin\pi(\theta_i-\theta_l)}|\\
\geq&2(-C\ln q_n -(q_n-1)\ln2)+\ln \|q_{n}\alpha\| \\
\geq&2q_n(-\ln2-\epsilon)-\epsilon q_{n} \\
>&2q_n(-\ln2-2\epsilon).
\end{align*}
Therefore we get
\begin{equation}\label{tan}
\prod_{l\neq i} \frac{|\sin\pi(x-\theta_l)|}{|\sin\pi(\theta_i-\theta_l)|} \leq e^{3h\epsilon}.
\end{equation}

\subsection{Case 2:}

$$q_n \leq k < \frac{1}{25}q_{n+1},\ h=2sq_n$$
We divide the $2sq_n$ points into $2s$ intervals: $T_1, \cdots, T_{2s}$, each containing $q_n$ points.
Fix any $i$. Let $|\sin\pi(\theta_i-\theta_{l_j})|$ be the minimal one of $|\sin\pi(\theta_i-\theta_l)|$ in $T_j, j=1, \cdots, 2s$. Without loss of generality, assume $i \in T_{j_0}, 1\leq j_0 \leq s$. 
\

\

We again estimate the two parts  in (\ref{sin})  separately. Using
(\ref{cosln2control})  and Lemma \ref{lana}:
\begin{align*}
       &\sum_{l\neq i} \ln|\sin\pi(x-\theta_l)|\\
 \leq&(C\ln q_n-(q_n-1)\ln 2)+(2s-1)(\epsilon-q_n\ln 2) \\
 \leq&2sq_n(-\ln 2+\epsilon).
\end{align*}
And
\begin{align*}
 & \sum_{l\neq i} \ln|\sin\pi(\theta_i-\theta_l)| \\
\geq&2s(-C\ln q_n-(q_n-1)\ln2)+\sum_{j=1,j\neq j_0}^s \ln|\sin\pi(\theta_i-\theta_{l_j})|
          +\sum_{j=s+1}^{2s} \ln|\sin\pi(\theta_i-\theta_{l_j})| \\
\geq&2sq_n(-\ln 2-\epsilon)+ I+II, \\
\end{align*}
where we set $I= \sum_{j=1,j\neq j_0}^s \ln\|(i-l_j)\alpha\|$ and
$II=\sum_{j=s+1}^{2s} \ln\|(i-l_j)\alpha\|$. Note that in the upper
bound it is enough to use Lemma \ref{lana} in each term, leading immediately to a
bound by $2sq_n(-\ln 2+\epsilon)$ but we present the estimate
the way we do, for clarity.

For $I$, the maximum distance between $i$ and $l_j$ is $sq_n$, which is clearly smaller than $k$, thus than $q_{n+1}$. Therefore by (\ref{qnalpha}), for large $|k|$,
\begin{equation}
I   = \sum_{j=1,j\neq j_0}^s  \ln\|(i-l_j)\alpha\|  \geq (s-1) \ln \|q_n \alpha\| \geq -sq_n\epsilon.
\end{equation}

For $II$, the maximum distance between $i$ and $l_j$ is $(k+2sq_n)$, which is smaller than $\frac{3}{25}q_{n+1}$, thus than $q_{n+1}$. Therefore we also have

\begin{equation}
II=\sum_{j=s+1}^{2s}  \ln\|(i-l_j)\alpha\|  \geq s\ln \|q_n \alpha\| \geq -sq_n\epsilon.
\end{equation}

Combining all the estimates above together, we get

\begin{equation}\label{tan1}
\prod_{l\neq i} \frac{|\sin\pi(x-\theta_l)|}{|\sin\pi(\theta_i-\theta_l)|} \leq e^{3\epsilon h}
\end{equation} as desired.
 $\hfill{} \Box $

\section{Proof of Lemma \ref{averagelower}}\label{Proof of Lemma averagelower}
\begin{proof}
We have
\begin{align*}
\tilde{P}_k(2\theta)=
\det{
\left[
\begin{array}{cccccc}
   t_0                                                         &                 c_0                                         &             & &    \\
  c_1                                                         &                 t_1                                         &  c_1                    \\
                                                                 &                 c_2                                        & \cdots   \\
                                                                  &                                                               &             & \cdots  &  c_{k-2} \\
                                                                  &                                                               &             & c_{k-1} &  t_{k-1}
\end{array}\right]_{k\times k}}
\end{align*}
where $t_j \triangleq E\cos{2\pi(\theta+\frac{j}{2}\alpha)}-\lambda\sin{2\pi(\theta+\frac{j}{2}\alpha)}$ and $c_j \triangleq -\cos{2\pi(\theta+\frac{j}{2}\alpha)}$.
Denote $z=e^{2\pi i \theta}$. Then
\begin{align}
\left\lbrace
\begin{matrix}
\tilde{t}_j(z)\triangleq e^{\pi i j\alpha}z\cdot t_j(z)=&\frac{E+i\lambda}{2}e^{2i\pi j\alpha}z^2+\frac{E-i\lambda}{2},\\
\tilde{c}_j(z)\triangleq e^{\pi i j\alpha}z\cdot c_j(z)=&-\frac{1}{2}e^{2i\pi j\alpha}z^2-\frac{1}{2}.
\end{matrix}
\right.
\end{align}
Since $|z|=1$, we have
\begin{align*}
|\tilde{P}_k(2\theta)|=
|f_k(z)|=
\left|\det{
\left[\begin{array}{cccccc}
   \tilde{t}_0(z)                                           &                 \tilde{c}_0(z)                          &             & &    \\
  \tilde{c}_1(z)                                           &                 \tilde{t}_1(z)                           &  \tilde{c}_1(z)                    \\
                                                                 &                 \tilde{c}_2(z)                           & \cdots   \\
                                                                  &                                                               &             & \cdots  &  \tilde{c}_{k-2}(z) \\
                                                                  &                                                               &             & \tilde{c}_{k-1}(z) &  \tilde{t}_{k-1}(z)
\end{array}\right]_{k\times k}
}\right|
\end{align*}
Clearly, $\ln{|f_k(z)|}$ is a subharmonic function, therefore 
\begin{align}\label{Phgeqf0}
\frac{1}{k}\int_{\mathbb{T}} \ln{|\tilde{P}_k(2\theta)|} \mathrm{d}\theta=\frac{1}{k}\int_{\mathbb{T}} \ln{|f(e^{2\pi i \theta})|} \mathrm{d}\theta\geq \frac{1}{k}\ln{|f_k(0)|}.
\end{align}
\begin{align}\label{f0=dh}
f_k(0)=&
\det{
\left[\begin{array}{cccccc}
  {(E-i\lambda)}/{2}                               &                -{1}/{2}                          &                   & &    \\
 -{1}/{2}                                               &                {(E-i\lambda)}/{2}            & -{1}/{2}                  \\
                                                                 &                -{1}/{2}                           & \cdots   \\
                                                                  &                                                           &                  & \cdots       & -{1}/{2}   \\
                                                                  &                                                           &                  & -{1}/{2}      &  {(E-i\lambda)}/{2}  
\end{array}\right]_{k\times k}
}\\
=&\frac{1}{(-2)^k}
\det{
\left[\begin{array}{cccccc}
  i\lambda-E                                 &                1                       &                  &              &    \\
   1                                               &                i\lambda-E        & 1                \\
                                                    &                1                       & \cdots        \\
                                                    &                                         &                  & \cdots   & 1   \\
                                                    &                                         &                  & 1           &  i\lambda-E
\end{array}\right]_{k\times k}
}\triangleq \frac{1}{(-2)^k} d_k.\notag
\end{align}
Obviously $d_k=(i\lambda-E)d_{k-1}-d_{k-2}$. Thus 

\begin{equation}\label{dkx2}
|d_k|\sim C|x_2|^{k-1}\ \mathrm{as}\ k\rightarrow \infty,
\end{equation} 
where $|x_1|<1<|x_2|$  are solutions of the characteristic equation $$x^2-(i\lambda-E)x+1=0.$$
Therefore by (\ref{Phgeqf0}), (\ref{f0=dh}) and (\ref{dkx2}), we have \footnote{ From this point on the proof can also be
  easily finished by a direct computation of $x_2$ and using the explicit
  expression for $L(E)$ in \cite{fgp}.}
\begin{align}\label{Phgeqx2}
\lim_{k\rightarrow \infty}\frac{1}{k}\int_{0}^1 \ln{|\tilde{P}_k(\theta)|} \mathrm{d}\theta\geq \ln{|x_2|}-\ln{2}. 
\end{align}
Clearly, $x_2$ is also the larger (in absolute value) eigenvalue of the following constant matrix
\begin{align*}
D_{\infty}=
\left(
\begin{matrix}
i\lambda-E\ &-1\\
1                 &0
\end{matrix}
\right)
\end{align*}
and by (\ref{defLE}), $\ln|x_2|$ equals to $L(\alpha,D_\infty).$ 
Then by a simple argument in the proof of Theorem 5.3 of
\cite{maryland} (based on the continuity of the Lyapunov exponent and
quantization of acceleration), we have $\ln{|x_2|}=L(\alpha, D_{\infty})=L(E)$.
Thus by (\ref{Phgeqx2}) we have
\begin{align*}
\lim_{k\rightarrow \infty}\frac{1}{k}\int_{0}^1 \ln{|\tilde{P}_k(\theta)|} \mathrm{d}\theta\geq L(E)-\ln 2.
\end{align*}
$\hfill{} \Box$
\end{proof}

\section*{Acknowledgement}
This research was partially supported by
NSF DMS-1401204.
F.Y. would like to thank Rui Han for many useful discussions, and Qi Zhou for some suggestions.

\bibliographystyle{amsplain}

\

\address{Svetlana Jitomirskaya, szhitomi@math.uci.edu\\ 
Department of Mathematics,
University of California, Irvine, California, USA}
\

\

\address{Fan Yang, ffyangmath@gmail.com\\
Department of Mathematics, University of California, Irvine, California, USA\\
Current address: School of Math, Georgia Institute of Technology, Atlanta, Georgia, USA}
\end{document}